\documentclass{article}

\PassOptionsToPackage{normalem}{ulem}
\usepackage{ulem}

\usepackage{algorithm}
\usepackage{algpseudocode}
\usepackage{graphicx}
\usepackage{epsfig}
\usepackage{caption}
\usepackage{subcaption}
\usepackage{amsmath}
\usepackage{amssymb}
\usepackage{amsthm}
\usepackage{xparse}
\usepackage{color}
\usepackage{xcolor}

\newtheorem{theorem}{Theorem}
\newtheorem{lemma}{Lemma}
\newtheorem{Claim}{Claim}

\newtheorem{corollary}{Corollary}

\newtheorem{definition}{Definition}

\def\squareforqed{\hbox{\rlap{$\sqcap$}$\sqcup$}}
\def\qed{\ifmmode\squareforqed\else{\unskip\nobreak\hfil
\penalty50\hskip1em\null\nobreak\hfil\squareforqed
\parfillskip=0pt\finalhyphendemerits=0\endgraf}\fi}

\begin{document}
\title{Optimal Evacuation Flows on Dynamic Paths with General Edge Capacities}

\author{Guru Prakash Arumugam\thanks{Currently with Google Inc.  Email: {\tt guruprakash991@gmail.com. } Work done while studying at the Dept. of Computer Science and Engineering (CSE), Indian Institute of Technology Madras (IITM) and visiting Hong Kong University of Science and Technology (HKUST).} \and 
John Augustine\thanks{Dept. of Computer Science and Engineering, IIT Madras, Chennai, India. Email: {\tt augustine@iitm.ac.in.} Supported by the IITM New Faculty Seed Grant, the IITM Exploratory Research Project, and the Indo-German Max Planck Center for Computer Science (IMPECS).} 
\and Mordecai J. Golin\thanks{Dept. of Computer Science and Engineering, HKUST, Hong Kong. Email: {\tt golin@cse.ust.hk}. Research partially supported by Hong Kong RGC CERG grant 16208415.} 
\and Yuya Higashikawa\thanks{Dept. of Information and System Engineering, Chuo University and CREST, Japan Science and Technology Agency (JST), Japan. Email:{\tt higashikawa.874@g.chuo-u.ac.jp}.} 
\and Naoki Katoh\thanks{Dept. of Informatics, Kwansei Gakuin University and CREST, Japan Science and Technology Agency (JST), Japan.  Email: {\tt naoki.katoh@gmail.com}. Research partially supported by JSPS Grant-in-Aid for Scientific Research(A)(25240004).} 
\and Prashanth Srikanthan\thanks{Currently with Microsoft Inc. Email: {\tt prashanthxs@gmail.com}. Work done while studying at the Dept. of Computer Science and Engineering (CSE), Indian Institute of Technology Madras (IITM) and visiting Hong Kong University of Science and Technology (HKUST).}}

\maketitle
\begin{abstract}
A  {\em Dynamic Graph Network} is a graph in which each edge has an associated travel time and a capacity (width) that limits the number of items that can travel in parallel along that edge. 
Each vertex in this dynamic graph network begins with the number of items that must be evacuated into designated sink vertices. 
A $k$-sink evacuation protocol  finds the location of $k$ sinks and associated evacuation movement protocol that allows evacuating all the items to a sink in minimum time.  The associated evacuation movement must impose a  {\em confluent} flow, i.e, all items passing through a particular vertex exit that vertex using  the same edge. 
In this paper we address the $k$-sink evacuation problem on  a dynamic path network.  
We provide solutions that  run  in $O(n \log n)$ time for $k=1$ and $O(k n \log^2 n)$ for $k >1$  and  work for  arbitrary edge capacities.
\end{abstract}


\newpage

\section{Introduction}
Dynamic Graph Networks were introduced by Ford and Fulkerson \cite{Ford1958} in 1958.  In these graphs, in addition to a weight denoting travel time, edges also possess a {\em capacity} denoting how much flow can enter the edge in one time unit. If too much flow is already at or entering  a vertex, the flow must then wait to enter the edge, causing  {\em congestion}.    Assuming a source $s$ and a sink $t$, canonical problems addressed in Dynamic Graph Networks are the {\em Max Flow over Time (MFOT)} problem   of how much flow can be moved from $s$ to $t$ in a given time $T$ and the {\em Quickest Flow Problem (QFP)} of how quickly $W$ units of flow can be moved from $s$ to $t$.  Good surveys of the area and applications can be found in
\cite{Skutella2009,Aronson1989,Fleischer2007,Pascoal2006}.  Ford and Fulkerson \cite{Ford1958} gave a polynomial time algorithm for the MFOT problem in their original paper;  strongly polynomial algorithms for  the QFP problem,   using techniques from parametric search, were developed later \cite{Burkard1993}.  Hoppe and Tardos \cite{Hoppe2000} give a strongly polynomial solution to a generalization, the {\em Quickest Transshipment Problem}, which has multiple sources and sinks, each with specified supplies and demands.  The problem there is to move items from the sources to sinks, satisfying the demands, as quickly as possible.

Dynamic Graph Networks can also model {\em evacuations}, e.g., \cite{Baumann2006}.  For example, supplies on a vertex  could  represent the number of people in a building or room;  edges represent hallways or roads  (with their travel times and capacity constraints); and sinks represent  emergency exits. The problem is to find an  evacuation protocol (describing who travels along which edge) that minimizes the  {\em evacuation time}  required for everyone to reach a sink. A more general problem would be, given $k$, to find the locations of the $k$ sinks and the associated evacuation protocols that minimize evacuation time over all possible locations of $k$ sinks.  Note that if $k >1,$ this differs from the Quickest Transshipment Problem  in that  demands on sinks, i.e., number of people evacuating to each exit, are not specified in advance but are calculated as part of the solution.

{\em Evacuation Flows} assume the equivalent of a sign in each room specifying {\em ``This Way Out''}. This restricts the above set-up by requiring that an evacuation protocol specify a unique {\em evacuation edge} for each vertex that all flow passing through that vertex  must unsplittably follow.  Such a flow is known as {\em confluent}\footnote{As described in \cite{Dressler2010b}, confluent flows occur naturally in problems other than evacuations, e.g.,  packet forwarding and railway scheduling};  even  in the static case constructing  an optimal confluent flow is known to be very difficult.  Not only is finding  a max-flow among all confluent flows  NP hard  but, if P $\not=$ NP, then it is even impossible to construct a constant-factor approximate optimal flow in polynomial time \cite{Chen2007,Dressler2010b,Chen2006,Shepherd2015}.
Note that confluent flows have  very particular forms.   Since following the evacuation edges from any vertex leads to a sink, the set of directed evacuation edges form a set of inwardly directed trees, with the ``root'' of the tree being the sink and the trees forming  a partition of the graph vertices.  If the $k$ sinks are provided as part of the input, the problem is to find the minimum evacuation time protocol for those sinks.   If only $k$ is given,  the sink-location version of the problem is to find the set of $k$ sinks along with associated evacuation protocol, that minimizes the evacuation time.   

 If  capacities are ``large enough'', the $k$-sink location problem for evacuation flows also becomes the $k$-center problem and is thus NP-Hard \cite[ND50]{garey1979computers}.  Kamiyama {\em et al.}~\cite{Kamiyama} proves by reduction to {\em Partition} that, even after fixing  $k=1$ and the location of the sink,  finding the min-time evacuation protocol is still NP-Hard. Alvarez and Serna  \cite{2015SANP} has a variation of this reduction that shows that the $k=1$ problem remains NP hard even when  all of the edge capacities are the same (uniform) and that this reduction works both when the sink location is known or is allowed to be calculated as part of the problem.
 
 Similar to the situation in the (static) $k$-medium and $k$-center problems,  
 e.g., \cite{frederickson1991optimal,frederickson1991parametric,woeginger2000monge,kariv1979algorithmicII,kariv1979algorithmicI,tamir1996pn},
 where the problem on general graphs is known to be NP-Hard, there is a strand of the literature,   that tries to find the most efficient solution for solving the problem  in special graphs for which polynomial algorithms can be found.  

In the evacuation flow problem, recent work has shown polynomial time algorithms for special cases of paths and trees.  More specifically, Mamada {\em et al}~\cite{Mamada2006} showed that for trees, the $1$-sink location problem can be solved in $O(n \log^2 n)$ time using a technique that essentially simulates the evacuation flows. 
The same authors also showed \cite{Mamada2005} how, if the locations of the $k$ sinks were given, the best evacuation protocol could be found in $O(k n^2 \log^2 n)$ time. The problem solved there was to {\em partition} the tree into $k$ subtrees, each containing one of the specified sinks, minimizing the maximum evacuation time over all the subtrees.
Higashikawa {\em et al}~\cite{Higashikawa2014b} showed how to solve the $k$-sink location  problem on a path  in $O(kn)$ time if all of the edge capacities are identical.  That paper used specific properties of uniform capacity dynamic graphs and therefore was unable to be extended to solve the $k$-sink location problem for general capacity paths.  A main result in this paper is the  explicit formula (rather than simulation) in Section \ref{sec:general_formula} for the evacuation time of a path.  This formula  permits deriving for a path,  an $O(n \log n)$ algorithm for the $1$-sink location problem and a $O(k n \log^2 n)$ time algorithm for the $k$-sink location problem.

\medskip

{\em \small Note 1:  The Dynamic Graph literature distinguishes between two models, {\em discrete} and {\em continuous}.  In the discrete model, flow is composed of 
unsplittable items, e.g., people to be evacuated.  This constrains the other parameters such as capacities, to be integral as well.  In the continuous  model, the flow is an infinitely  splittable item such as water.  Capacity is then 
no longer constrained to be integral, i.e., it represents the rate at which flow can enter an edge.  As noted in \cite{Fleischer1998}, many techniques that work in the discrete  model can be modified to work in the continuous model as well.  This paper will assume the discrete model but the algorithms we present can also be modified (using methods akin to those in \cite{Fleischer1998}) to work within the same time bounds in the continuous case as well.}

{\em \small Note 2: The algorithms in this paper extend the general model to allow the sinks to be placed anywhere on the path and not just on vertices.  They can easily and straightforwardly be restricted to only permit sink placement on vertices with the same running times.
}


\section{Model Definition and Results}

\begin{figure}[t]
\begin{center}
\includegraphics[scale=0.6]{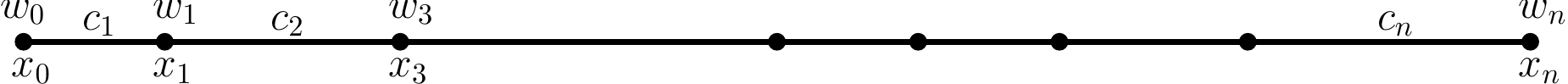}
\end{center}
\caption{The problem's  input specifies   vertices  $x_0 < x_1<\cdots < x_n$, the number of people $w_i$ starting at vertex $v_i$ and the capacity $c_i$  of the edges $e_i = (x_{i-1},x_i).$}
\label{fig:Basic_Input}
\end{figure}

Consider a path $P=(V,E)$ with $(n+1)$ vertices (buildings)
$V=\{x_{0},x_{1},...,x_{n}\}$ and $n$ edges (roads) 
$E=\{e_{1},e_{2},...,e_{n}\}$;  $e_{i}=(x_{i-1},x_{i})$. We also use $x_{i}$ to denote the (line coordinate)
location of the 
$i^{th}$ building and $x_0 < x_1 < ... < x_n$.  For $i <j$,  let $P_{i,j}$ denote the subpath $x_i < x_{i+1} < ... < x_{j}$.  The size of $P_{i,j}$ is the number of its vertices, i.e.,  $|P_{i,j}|= j-i+1$. Thus,
$|P|=n+1.$ See Figure \ref{fig:Basic_Input}.

Each vertex $x_{i}\in V$ has an  associated weight $w_i$ ($\ge 0$), denoting the number of evacuees (or simply, just {\em people})  in building $x_i$. All of the  evacuees on vertex $x_j$ must evacuate to some (the same) designated sink $x\in P$.
The distance between vertices $x_{i}$ and $x_{j}$ is 
$|x_{i}-x_{j}|$.  We use $\tau$ to denote the time required to travel one unit of distance.
Travel times between vertices, i.e, $|x_{i}-x_j|\tau$ are all assumed  to be integral.
Each edge $e_i$ also has an integral  capacity $c_i$, equal to the maximum number of evacuees who can enter $e_i$  in unit time.


It is convenient to interpret the capacity as the {\em width}  of the edge, i.e., the number of people who can travel together in parallel along the edge.  An  evacuation flow consists of \emph{waves} of evacuees moving through the edge. The size of a \emph{wave} is bounded by the capacity of the edge through which it is passing.
Each \emph{wave} covers a unit distance on the edge in time $\tau$. People wanting to travel an edge  will line up and wait at the entrance to the edge until they can join a wave to pass through it.


\begin{figure*}[t]
\centering
\hspace*{-0.75in}
\vspace*{.1in}
\begin{subfigure}{0.29\textwidth}
\includegraphics[scale=0.4]{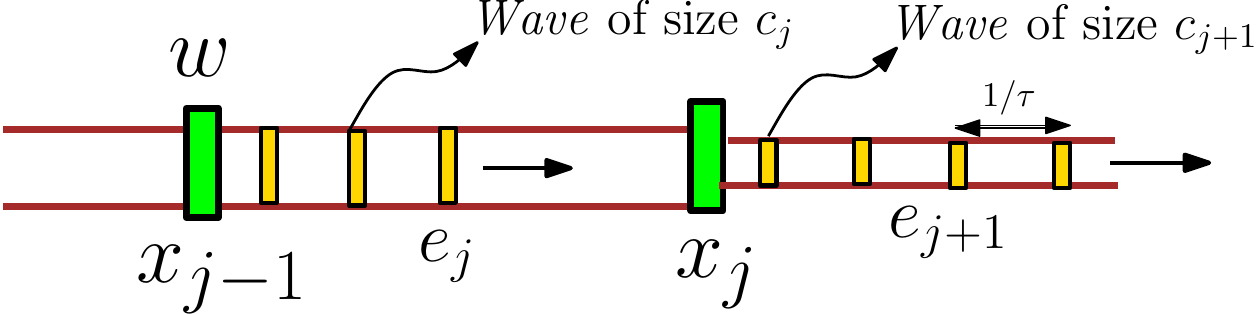}
\caption{\label{fig:waves} Waves of flow}
\end{subfigure}
\hspace{1in}
\begin{subfigure}{0.4\textwidth}
\includegraphics[scale=0.4]{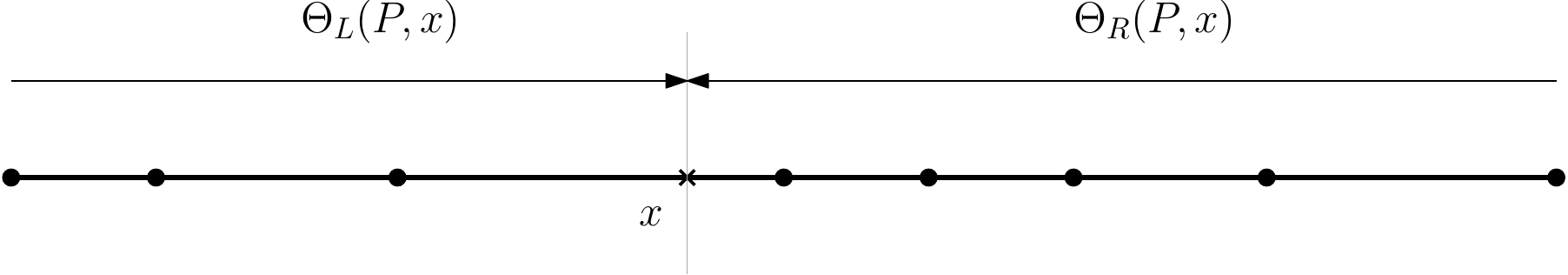}
\caption{\label{fig:1_sink_evac}$1$-sink evacuation}
\end{subfigure}
\caption{\label{fig:basic:evac} (a) illustrates how flow moves between edges.  $c_j$ people move from $x_{j-1}$ to $x_j$ in waves that are one time unit apart. The first wave reaches $x_j$ at time $\tau(x_{j} - x_{j-1})$ and the last at time $\tau(x_{j} - x_{j-1})+  \lceil w/c_{j}\rceil -1$. (b) illustrates how items in path $P$ evacuate to sink $x$.  $\Theta_L(P,x)$ and $\Theta_R(P,x)$ are, respectively, the times for the items to the left/right of $x$ to fully evacuate to $x$.}
\end{figure*}

Consider the simple situation in which $w$ people on  $x_{j-1}$ move to $x_j$ over edge $e_j$ (Figure \ref{fig:waves}).
At time $t=0$, the first $c_j$ people enter $e_j$ as the first wave; at $t=1$, the 2nd $c_j$ of them enter  $e_j$ as the 2nd wave,  etc. There are $\lceil w/c_j\rceil$ waves in total with all but possibly the last containing $c_j$ items.  The first wave reaches $x_j$ at $t=(x_j-x_{j-1})\tau$, the second at $t=(x_j - x_{j-1})\tau +1$, etc.  The last item 
reaches   $x_j$ at time $t=(x_j - x_{j-1})\tau + \lceil w/c_j\rceil -1.$

{\em Congestion} can occur in one of two ways.  In the first (Figure \ref{fig:2sinka}),  people travelling together in a wave reach a new edge whose capacity is smaller than the size of the wave, i.e., $c_{j+1} < c_j.$  In the second (Figure \ref{fig:2sinkb}), people reach a vertex at which too many other people are already waiting. In both cases, people will need to wait before entering the new edge. 





Given a path $P$ and a single sink $x\in P$,  {\em evacuation} consists of all people in the buildings moving to $x$. Note that $x$ is not required to be situated on a vertex but if it is, then all items starting at  $x$ are considered to be immediately evacuated  (in time $0$).  People to the left and right of the sink $x$ need to evacuate to $x$ (as shown in Figure \ref{fig:1_sink_evac}).   We assume by convention that items at a vertex  wait in a queue, with earlier arrivals leaving the vertex first.  
A major issue in the analysis of evacuation protocols is dealing with the congestion that occurs when items  have to wait at a vertex before leaving.  

\begin{figure*}[ht]
\vspace*{.3in}
\begin{subfigure}[b]{0.36\textwidth}
\includegraphics[width=2.2in]{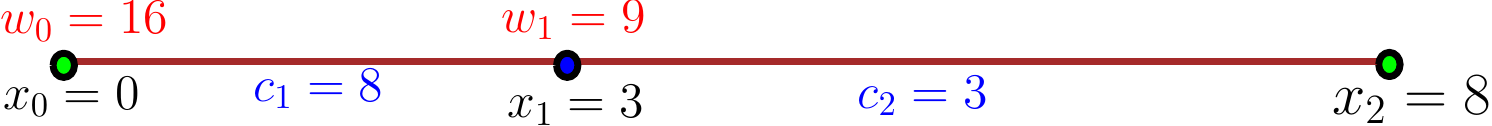}
\caption{\label{fig:2sinka}  $c_1 > c_2$}
\end{subfigure}
\hfill
\begin{subfigure}[b]{0.36\textwidth}
\includegraphics[width=2.2in]{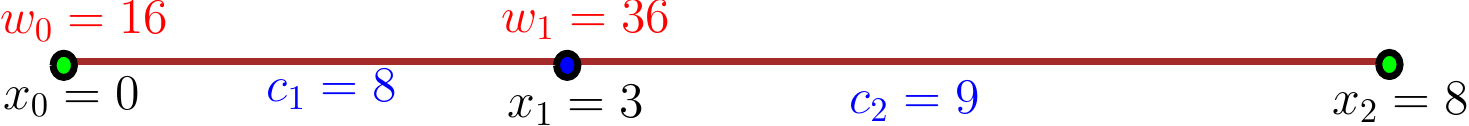}
\caption{\label{fig:2sinkb}$c_1 \le c_2$}
\end{subfigure}
\caption{\label{fig:2sink} In both  cases  $\tau=1$,  all items at $x_0$ and $x_1$ are evacuating to $x_2$  and   congestion occurs at $x_1$,   but for different reasons.  In both cases the first $c_1=8$  items from $x_0$ reach $x_1$ at time $t=3$ and the second 8 items at time $t=4.$  In (a),  because $c_1 > c_2$, not all of the first 8 items reaching $x_1$ can pass through $x_1$ immediately. It isn't until $t=13$ that the last items finally reach $x_2.$   In (b), $c_1\le c_2$ but congestion still occurs at $x_1$ because, at time $t=3$, 9 of the original items at $x_1$ still remain there and the new items from $x_0$ have to wait until they leave. It isn't until $t=12$ that the last items finally reach $x_2.$}
\end{figure*}


Before starting, we introduce some notation that allows us to formulate and state the results.
\begin{definition}\ 
\vspace*{-.1in}
\begin{itemize}
\item   $\Theta_L(P,x)$(resp. $\Theta_R(P,x)$)  is the time required   for all people to the left(resp. right) of sink $x$ to evacuate to $x$. 
\item Equivalently,  $\Theta_L(P,x)$(resp. $\Theta_R(P,x)$) is the time required for the last item from $x_0$ (resp $x_n)$ to reach $x$
\end{itemize}
\end{definition}





\vspace*{-.05in}
One of the  main technical results of this paper is the following theorem, proven in Section  \ref {sec:general_formula}.

\begin{theorem}
\label{thm:main form}
Let $k$ be such that $x_k < x \le x_{k+1}$. Then,
\begin{eqnarray}
 \Theta_{L}(P,x) &=& \max_{x_i<x}  \left( (x-x_{i})\tau   + \left\lceil\frac {\sum_{0\leq j\leq i}w_{j}}{\min_{i+1\leq j\leq k+1} c_{j}}\right\rceil  - 1 \right)\label{eq:left-evac-gc}\\
 \Theta_{R}(P,x) &=& \max_{x_i>x}  \left( (x_{i}-x)\tau  + \left\lceil\frac {\sum_{i\leq j\leq n}w_{j}}{\min_{k+1\leq j\leq n} c_{j}}\right\rceil  - 1 \right)\label{eq:right-evac-gc}
\end{eqnarray}
\end{theorem}

The specialization of these equations to the uniform capacity case  where 
$\forall i, c_i =c$ for some fixed $c>0$, was given in \cite{Kamiyama2009a} 
The proof there was very specific to the uniform capacity case and could not be extended to general $c_i$.

\medskip

{\em \small Note: As mentioned in the introduction, this paper assumes {\em discrete} flows but the results can be extended to continuous flows.  In the continuous flow case (\ref{eq:left-evac-gc}) and  (\ref{eq:right-evac-gc}) 
are replaced by
\begin{eqnarray}
 \Theta_{L}(P,x) &=& \max_{x_i<x}  \left( (x-x_{i})\tau   + \frac {\sum_{0\leq j\leq i}w_{j}}{\min_{i+1\leq j\leq k+1} c_{j}} \right)\label{eq:left-evac-gcc}\\
 \Theta_{R}(P,x) &=& \max_{x_i>x}  \left( (x_{i}-x)\tau  + \frac {\sum_{i\leq j\leq n}w_{j}}{\min_{k+1\leq j\leq n} c_{j}} \right).\label{eq:right-evac-gcc}
\end{eqnarray}
See \cite[p. 32]{Higashikawa2014c} for more details.
}

\medskip

Intuitively, for fixed $t$, the value on the right hand side of (\ref{eq:left-evac-gc}) is the time required for evacuation to $x$ if (i) there were no people to the right of $x_i$ and (ii) all the people on or to the left  of $x_i$ were originally located {\em on} $x_i$,  with a similar intuition for (\ref{eq:right-evac-gc}). The proof in Section \ref {sec:general_formula}  formalizes this intuition.

\begin{definition}
Define 
$$\Theta(P,x)=\max\left\{ \Theta_{L}(P,x),\Theta_{R}(P,x)\right\},
\quad \quad 
\Theta^1(P) =  \min_{x \in P} \Theta(P,x)$$
as, respectively, 
 the maximum time required to evacuate all people in $P$　\ to $x$.
and  the minimum time required to evacuate $P$ (to an optimally chosen sink).
\end{definition}

For fixed $x$, Theorem \ref{thm:main form} automatically gives  an $O(n)$ algorithm for constructing $\Theta(P,x)$.  First preprocess by using three  linear time scans to  calculate 
$W_i = \sum_{j=1}^i w_i$, $c'_i = \min_{0 < j \le i} c_i$  and $c''_i = \min_{i \le j \le n} c_j$.  Once these are known, $\Theta_L(P,x)$ and $\Theta_R(P,x)$ and thus $\Theta(P,x)$ can be calculated in another $O(n)$ time.  We state this as a claim for later reference

\begin{Claim}
\label{claim:theta1PXalg}
Given $x \in P$, $\Theta_L(P,x),$ $\Theta_R(P,x)$ and $\Theta(P,x)$ can be calculated in $O(|P|)$ time.
\end{Claim}
Section~\ref{sec:1_sink_evac} shows how to use  Claim~\ref{claim:theta1PXalg}  to compute the {\em minimum} 1-sink evacuation time:

\begin{theorem}
\label{thm:1sink}
Given $P,$  there is an  $O(|P| \log |P|)$ time algorithm for calculating 
$\Theta^1(P)$ and $x^*$ such that
$$\Theta^1(P) = \max(\Theta_L(P,x^*),\Theta_R(P,x^*)).$$
\end{theorem}

A {\em $k$-sink evacuation protocol}  for a path $P$ requires splitting $P$ into $k$ subpaths and assigning a different sink to each subpath.
More explicitly (see Figure \ref {fig:k_sink_partition_a}), (i) let $I=(i_0,i_1\ldots,i_k)$ be an increasing sequence of indices $0=i_0<i_1<\cdots i_k=n$ that partition $P$ into $\hat{P}^I= \{P^I_1,P^I_2,\ldots,P^I_k\}$ where $P^I_j= P_{i_{j-1},i_{j+1}-1}$ 
and (ii) let the sinks be specified by $Y=(y_1,y_2,...,y_k)\in {\cal Y}^I$  where
${\cal Y}^I = \{(y_1,y_2,\ldots,y_k) \,:\, y_j \in P^I_j\}$.  The time required for the evacuation protocol 
for partition $I$ with sinks $Y$ 
is the maximum evacuation time for the $k$ subpaths.
$$\Theta^k(P,\{I,Y\}) = \max_{1 \le j \le k} \Theta^1(P^I_j,y_j)$$
 Let $\cal I$ be the set of all legal partitions of $P.$  Then
\begin{definition}
\label{def:ksink1}
An optimal $k$-sink evacuation protocol is one that has minimum value over all evacuation protocols, i.e., 
$$\Theta^k(P) = \min_{I \in {\cal I}} \max_{Y \in {\cal Y}^I} \Theta^k(P,\{I,Y\}).$$
\end{definition}

\begin{figure}[t]
\begin{center}
\includegraphics[scale=0.6]{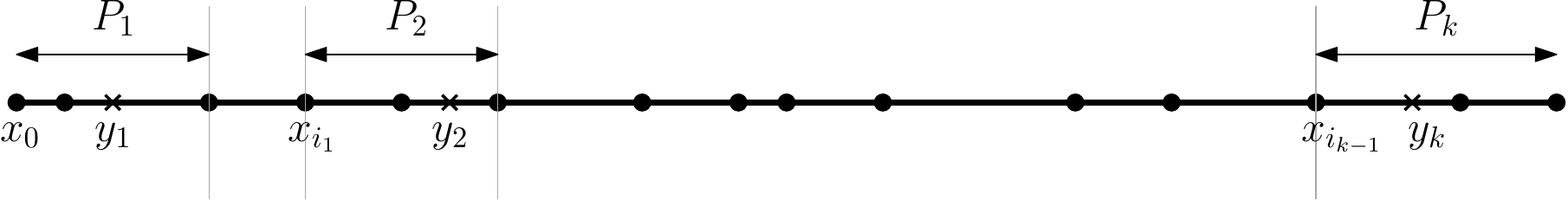}
\end{center}
\caption{In $k$-sink evacuation $P$ is partitioned into $k$ subpaths  $[x_{i_j},x_{i_j+1}-1]$ with each subpath having sink $y_j$.  The evacuation time is $\max_j \Theta(P_{i_j},y_j)$.,}
\label{fig:k_sink_partition_a}
\end{figure}

Section \ref{subsec:k_sink_evac_Test}   bootstraps off of Theorem \ref{thm:1sink} to prove
\begin{theorem}
\label{thm:ksink_ver}
Let $k \ge 0 $ be an integer and $\alpha >0$ a real. Then, 
\begin{equation}
\Theta^k(P) \le \alpha
\end{equation}
can be verified in $O(|P| \log |P|)$ time.  If $\Theta^k(P) \le \alpha$, the verification procedure will also output 
a $k$-sink evacuation protocol for $P,$ with evacuation time $\le \alpha$,  i.e., a partition $I$ and sinks $Y \in {\cal Y}^I $ such that
 $$\max_{1 \le j \le k} \Theta^1(P^I_j,y_j).$$
\end{theorem}
Finally, Section \ref {subsec:k_sink_evac_Alg} combines Theorems \ref{thm:1sink} and \ref {thm:ksink_ver} to prove the final result.
\begin{theorem}
\label{thm:ksink}
There is an $O(k |P| \log^2 |P|)$ algorithm for calculating $\Theta^k(P)$.  This algorithm will also
provide a $k$-sink evacuation protocol for $P,$  i.e., a partition $I$ and sinks $Y \in {\cal Y}^I $ such that
\begin{equation}
\Theta^k(P) =  \max_{1 \le j \le k} \Theta^1(P^I_j,y_j)
\end{equation}
\end{theorem}




\section{The 1-sink Evacuation Time Formula}
\label{sec:general_formula}
The goal of this section is to prove Theorem \ref{thm:main form}.
Since the expressions for $\Theta_L(P,x)$ and $\Theta_R(P,x)$, the  left and right evacuation times,    will be symmetric,   we only prove the formula for $\Theta_L(P,x)$.   That is, we derive a formula for calculating the evacuation time of all nods on a  path to a sink on the left of the path.  This formula will be evaluable in $O(n)$ time.

The best previously known  for calculating the evacuation  time was  $O(n \log^2 n)$  using the algorithm for a tree given in \cite{Mamada2006}.  Note that this algorithm did not evaluate a formula.  It essentially simulated the movement of the flows using clever data structures.  For a path, a formula was known
\cite{Kamiyama2009a} for uniform capacities but the proof for that case did not extend to general capacities. \footnote{
Our proof of Theorem \ref{thm:main form} follows from a first principal examination of  the actual structure of the flow.  One of the main tools for attacking single source single sink dynamic flows is a variation, originally introduced by Ford and Fulkerson  \cite{Ford1958}  (see the survey of \cite{Skutella2009} for more details)  of the max-flow min-cut theorem  for static flows.  It is an interesting open question as to whether the formula of Theorem \ref{thm:main form} could be derived using similar cut techniques.}

To prove  Theorem \ref{thm:main form} first  consider a dynamic path $P$ with edge capacities $c_i$ and  a sink $x\in P$. Let $P_L$ be the path to the left of sink $x\in P$. 
(Figure \ref{fig:gen-capacity})  Let $x_r$ be the the rightmost vertex satisfying $x_r <x$ (if $x$ is on a vertex, $x=x_{r+1}$. Let $W_i= \sum_{0 \le j \le i} w_j$  be the number of people on or to the left of node $i$.
For our proofs we assume that the $W_r$ people who are evacuated are labelled consecutively from left to right, with the $w_i$ people on node $i$ labelled arbitrarily. That is, the people on vertex $x_i$ are labelled from $W_{i+1}$ to 
$W_i$. Without loss of generality our proofs  assume that people wait at each vertex in a first-in first-out queue and leave each vertex in   increasing labelled order.


\begin{lemma} 
\label{lem:capacity-change} Suppose there exists a vertex $x_j\in P_L$ such that 
$c_j > c_{j+1}$. Then the left evacuation time does not change if the  capacity of $e_j$ is set  to be  $c_{j+1}$.
\end{lemma}

\begin{figure*}[tp]
\begin{subfigure}[c]{0.4\textwidth}
\includegraphics[width=3in]{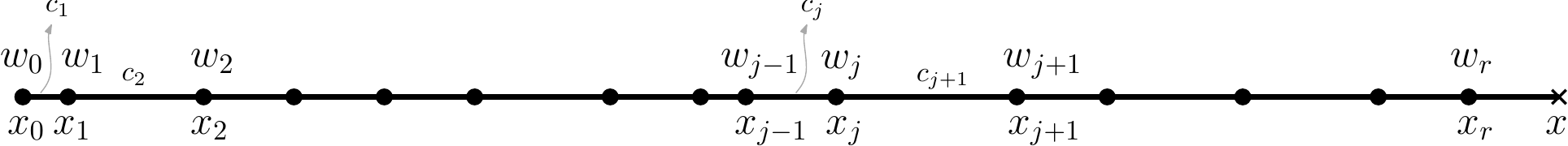}
\caption{\label{fig:gen-capacity}}
\end{subfigure}
\hfill
\begin{subfigure}[c]{0.40\textwidth}
\hspace*{.38in} \includegraphics[width=2in]{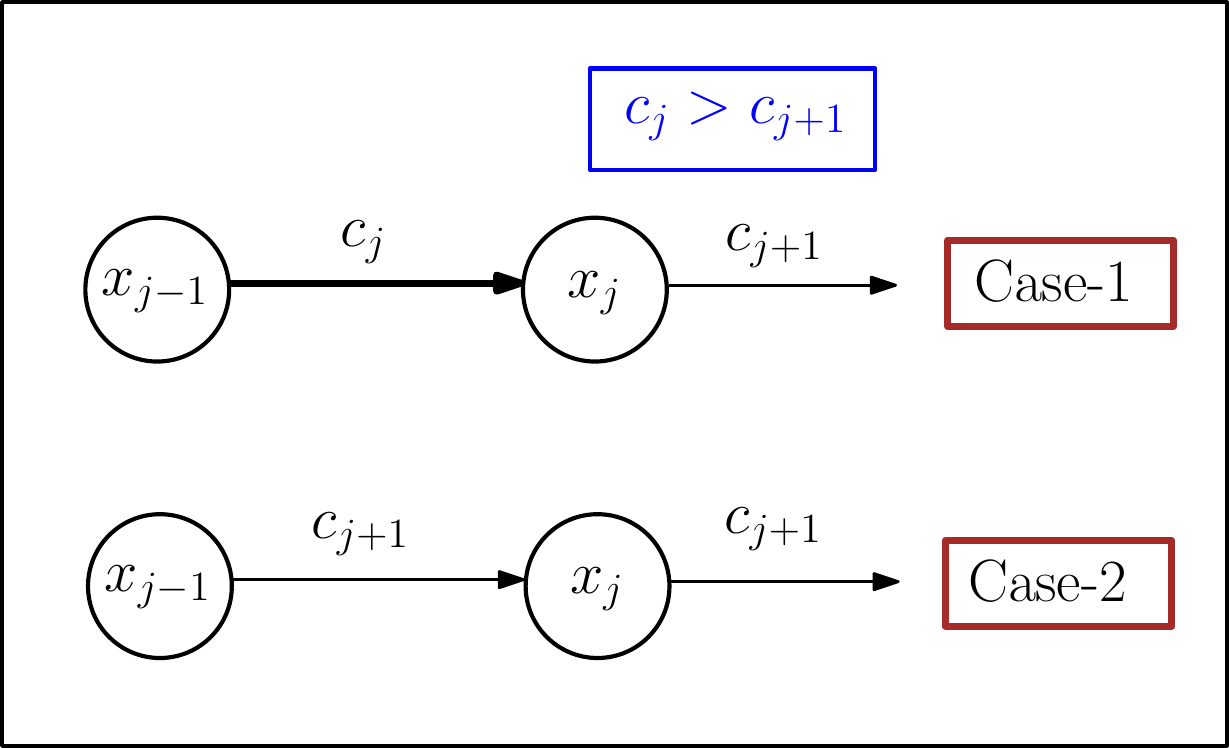}
\caption{\label{fig:gen-capacity-waves}}
\end{subfigure}
\caption{Left Evacuation. In (a) everyone is evacuating to the rightmost point $x.$
$x_r$ is the rightmost path vertex satisfying $x_r <x.$ (b) illustrates the two Cases in the proof of Lemma \ref{lem:capacity-change}.
 Case 1 has original capacities.  Case 2 replaces $c_j$ with $c_{j+1}$.}
\end{figure*}

\begin{proof} Refer to Figure \ref {fig:gen-capacity-waves}.
Define {\bf  Case  1} to   be the evacuation with  the original capacities and {\bf Case 2} to be  the evacuation when  $c_j$ is replaced by $c_{j+1}$.

The analysis requires notation to track the flow of people through edge $e_j$.
Since we will only be counting  people starting on or to the left of $x_j$, subtract $W_r-W_j$ 
from each label.  Person 1 is now the first person on or to the left of $x_j$ and person $W_j$ is the last person on the leftmost node, i.e., the last person that will pass through $x_j$.

People $\ell \le w_j$ start on $x_j$.  For Case $i$, $i=1,2$, for the nodes starting on $x_{j-1}$ or to its left, i.e., $\ell> w_j$, set 
$$
a'(\ell) =  \mbox{time item $\ell$ arrives at node $x_{j-1}$}, 
\quad
t'_i(\ell) =    \mbox{time  item $\ell$ leaves $x_{j-1}$}
$$
where we use the convention that if $\ell$ starts on $x_j$,  $a'(\ell) =0.$ 
Note that the path to the left of $x_{j-1}$ is the same in both cases, so  $a'(\ell)$ is the same in  both cases.  But, because the capacity of $e_j$ is different in the two cases, $t'_i(\ell)$ can be different for $i=1,2$.

 We first prove   by induction that
$\forall \ell >  w_{j} ,\quad t'_1(\ell)   \le  t'_2(\ell).$
This  is obviously true for $\ell = w_{j} +1,$ since the first item starting or passing through 
$x_{j-1}$ doesn't experience congestion at $x_{j-1}$ so 
$$ t'_2(w_{j} +1) = a'(w_{j} +1) = t'_1(w_{j} +1) .$$
Now suppose that 
$t_1'(\ell) \le t_2'(\ell)$    for all $\ell< \ell'$.  From the definitions, $s < \ell$ iff $s$ is {\em ahead of}  $\ell$, either on a vertex to the right of $\ell$ or on the same vertex as $\ell$ but ahead of it in the queue. Thus, in Case $i,$ the number of people waiting at $x_{j-1}$ at the time $\ell'$ arrives there is 
\begin{equation}
\label{eq:nidef}
n_i(\ell') = | \{s \,:\, s < \ell'  \mbox{ and }  a'(\ell') \le t'_i(s)\}|.
\end{equation}

 By the  induction hypothesis $n_2(\ell') \ge n_1(\ell')$ so $\ell'$ waits for at least as many items in case 2 as in case 1.  Th capacity of $e_j$ in Case 2 is no greater than the capacity in Case 1, so  $\ell'$ waits on $x_{j-1}$ at least as long in Case 2 as in Case 1,  proving $t'_1(\ell')   \le  t'_2(\ell')$  and thus the validity of $t'_1(\ell)   \le  t'_2(\ell)$ for all $\ell \ge  w_{j} +1.$
 
Now  set 
\begin{eqnarray*}
a_i(\ell) 
&=&
\left\{
\begin{array}{ll}
0 & \mbox{if $\ell \le w_{j+1}$}\\
t'(\ell) + (x_{j} - x_{j-1})\tau     & \mbox{if $\ell > w_{j+1}$}
\end{array}
\right.\\
&=& \mbox { the time that $\ell$ arrives at $x_{j}$}\\
 t_i(\ell)  &=&  \mbox{time that item $\ell$ leaves  node $x_{j}$}
\end{eqnarray*}
 
 Note that  $\forall \ell >  w_{j} ,\quad t'_1(\ell)   \le  t'_2(\ell).$ immediately implies
 \begin{equation}
\label{eq:a'_ind}
\forall \ell,\quad  a_1(\ell)  \le a_2(\ell)
\end{equation}

To prove the lemma it suffices to prove 
\begin{equation}
\label{eq:maint}
\forall \ell \le m,\quad   t_1(\ell) = t_2(\ell)
\end{equation}
i.e., the time that nodes leave  $x_{j}$  is  independent of the case.
Since the path from $x_{j}$ to the sink is the same for both cases, the final evacuation time can not differ in the two cases.

(\ref{eq:maint}) will also be proven by induction. Since the first $c_{j+1}$ items, regardless of where they start, do not experience any congestion waiting for $e_j$ or $e_{j+1}$
$$\forall \ell \le c_{j+1},\ t_1(\ell) = t_2(\ell).$$
Now assume that $t_1(\ell') = t_2(\ell')$ for all $\ell' \le \ell$.  There are three different situations in Case 2;  each of them will be shown to  imply $t_1(\ell+1) = t_2(\ell+1)$, proving the induction step.

\medskip

\par\noindent 
{\bf (a) \underline{$\bf a_2(\ell+1) \le t_2(\ell):$}\\[0.05in]
In Case 2,  $\ell+1$ arrives at $x_j$ before $\ell$ leaves $x_j$.\\[0.1in]}
Consider the situation  in Case 1 at $x_{j}$ at time $a_2(\ell+1).$\\  Because
$a_1(\ell+1) \le a_2(\ell+1)$,  $\ell+1$ has already arrived at $x_{j+1}$.  Because
$$a_1(\ell+1) \le a_2(\ell+1) \le t_2(\ell) = t_1(\ell)$$
$\ell+1$ has not left $x_{j}$. Thus  at time $a_2(\ell+1),$  in both Case 1 and Case 2,  $\ell+1$ is at $x_{j+1}$ and the set of items it is waiting for is identically
$$\{u < \ell +1 \,:\, t_1(u) \ge a_2(\ell+1)\}
=\{u < \ell +1 \,:\, t_2(u) \ge a_2(\ell+1)\}
$$
(that  these sets are equal follows  from the induction hypothesis). Since the sets are identical, 
in both cases $\ell+1$ waits the same further amount of time to leave $x_{j}$ and $t_1(\ell+1) = t_2(\ell+1)$.

\medskip

\par\noindent 
{\bf (b) \underline{$\bf a_2(\ell+1) > t_2(\ell)$ and $a'_2(\ell+1) =t'_2(\ell+1):$ }\\[0.05in]

In Case 2, $\ell+1$ did not experience congestion at $x_{j-1}$ and arrives at $x_j$ after
$\ell$ leaves $x_j$.\\[0.1in]}
Note that
$$a'(\ell +1) \le t'_1(\ell+1) \le t'_2(\ell+1) = a'(\ell+1)$$
where the first inequality is by definition, the second is from (\ref{eq:maint}) and the third is by assumption.   Therefore $t'_1(\ell+1) = t'_2(\ell+1)$ and thus $a_1(\ell+1) = a_2(\ell+1).$

From the induction hypothesis, 
$$a_1(\ell+1)=a_2(\ell+1) > t_2(\ell) = t_1(\ell)$$
so, in  both cases, when $\ell+1$ arrives at $x_{j}$ no other nodes are waiting at $x_{j}$.
Thus, in both cases, $\ell+1$ immediately leaves $x_{j}$ with no delay and
$$t_1(\ell+1) = a_1(\ell+1) = a_2(\ell+1) =  t_2(\ell+1).$$

\newpage

\par\noindent 
{\bf (c) \underline{$\bf a_2(\ell+1) > t_2(\ell)$ and $a'_2(\ell+1)  < t'_2(\ell+1):$ }\\[0.05in]
In Case 2, $\ell+1$ DID experience congestion at $x_{j-1}$ and arrives at $x_j$ after
$\ell$ leaves $x_j$.\\[0.1in]}
Consider the wave that transports $\ell+1$ from $x_{j-1}$ to $x_{j}$ in Case 2.
If that wave also  contained $\ell$, then $\ell$ and $\ell+1$ arrive at 
$x_{j}$ at the same time so
$$ a_2(\ell+1) = a_2(\ell) \le t_2(\ell)$$
contradicting the assumption $a_2(\ell+1) > t_2(\ell)$.

Thus $\ell$ and $\ell+1$ are in different waves.  Since $\ell+1$ waited at $x_{j-1}$, the wave containing $\ell+1$ arrives at $x_{j}$ exactly one time unit after the wave containing $\ell,$ i.e.,
$$a_2(\ell+1) = a_2(\ell) +1 \le t_2(\ell) +1$$
Thus, $a_2(\ell) = t_2(\ell)$ and $a_2(\ell+1) = t_2(\ell) +1.$

Furthermore, the only way that $\ell+1$ could not be in the same wave as $\ell$ in Case 2  is if the wave $W$ containing $\ell$ was full, i.e.,
$$W = \{\ell,\ell-1,\ldots,\ell-c_j+1\}.$$
Since all of these items arrive at $x_{j}$ at  time $a_2(\ell)$ and none of them are there at time $a_2(\ell+1)$ (if any of them stayed, then $\ell$ would have stayed) we must have that $t_2(\ell-j) = a_2(\ell)=t_2(\ell)$ for $j=0,1,\ldots, c_{j+1}-1$.

Now consider the behavior of $\ell+1$ in Case $1$. 
$a_1(\ell+1) \le a_2(\ell+1) = t_2(\ell)+1$, 
so $\ell+1$ is at $x_{j}$ by time $t_2(\ell) +1.$  But, by assumption, 
$t_1(\ell-j) = t_2(\ell-j) =t_2(\ell)$ for $j=0,1,\ldots, c_{j+1}-1$.  
Since those $c_{j+1}$ items are already leaving $x_{j}$ at time $t_2(\ell)$ over an edge of capacity $c_{j+1}$, $\ell+1$ can not join them.  So $t_1(\ell+1) = t_2(\ell)+1 = t_2(\ell+1)$ and we are done.
\end{proof}

\begin{figure}
\begin{center}
\includegraphics[width=5in]{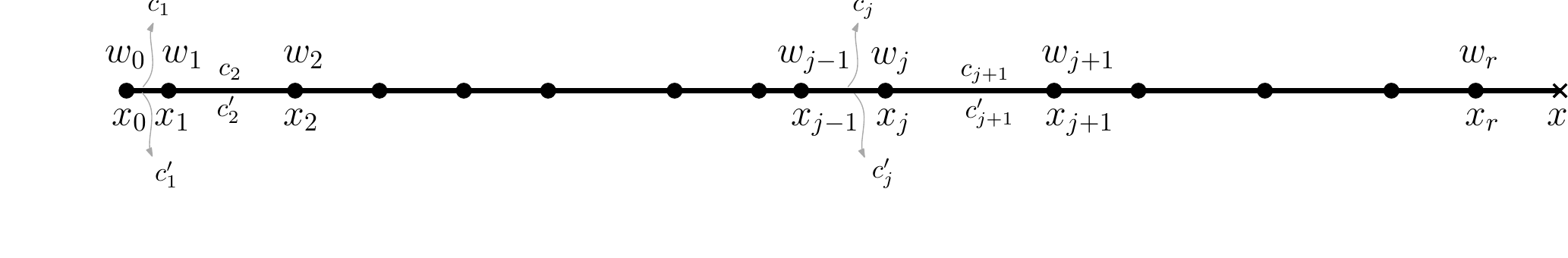}
\end{center}
\caption{All vertices are evacuating to the rightmost point $x$.  Path $P$ is the path with capacities $c_i.$  Path $P'$ is the same path but with the capacities replaced by $c'_i = \min_{i \le  j \le r+1} c_j$.  Corollary \ref {cor:eq-problem} states that $\Theta_L(P,x) = \Theta_L(P',X)$}
\label{fig:ThetaL_2}
\end{figure}

\medskip

Starting with the rightmost edge and walking left, applying Lemma \ref{lem:capacity-change} to each edge yields (See Figure \ref {fig:ThetaL_2})

\begin{corollary} \label{cor:eq-problem} Let $x$ be the location of the sink and $x_r$ the rightmost vertex to the left of $x$.  Let $P'$ be $P$ with new edge capacities
$c'_i = \min_{i \le  j \le r+1} c_j$, $i=0,1,\ldots,r$. Then
$$\Theta_L(P,x) = \Theta_L(P',x)$$
\end{corollary}

\begin{lemma}
Let $x$ be the location of the sink and $x_r$ the rightmost vertex to the left of $x$. 
Set $c'_t = \min_{t \le j \le r+1} c_j$ and  $W_t=\sum_{0\leq j\leq t}w_{j}$. Then
\begin{equation}
\Theta_L(P,x) = 
\max_{0 \le t \le r}
\left(\left((x-x_{t})\tau+\left\lceil \frac{W_t} {c'_{t+1}}\right\rceil -1\right)\right)
\end{equation}
\end{lemma}
\begin{proof}
From Corollary \ref{cor:eq-problem} we can assume that we are calculating $\Theta_L(P',x)$ that has
capacities 
$c'_t = \min_{t \le j \le r+1} c_j$ instead of the original $c_t.$

Recall the  labelled numbering of the people.  The very last person on $x_0$ is $L=W_r$.
$\Theta_L(P,x)$ is  the time taken by $L$   to reach $x$.

   For fixed $t$, set $T_t$ to be the time at which $L$ leaves $x_t$ (entering  $e_{t+1}$). After leaving $x_t,$ $L$ needs at least $(x-x_t)\tau$ more time to travel to $x$,  so $\Theta_L(P,x) \ge  T_t + (x-x_t)\tau	.$

    Each wave leaving $x_t$ has  size at most $c'_{t+1}$.  Since there are
    $W_t$ people  on nodes $0,\ldots,x_t$,  at least $\lceil W_t/c'_{t+1}\rceil$ waves leave $x_t$ with $L$ being in the last wave.  Only one wave leaves $j$ at each time step, so  $T_t \ge  \lceil W_t/c'_{t+1}\rceil-1$. Thus, the time at which $x$ reaches $x$ is at least 
    $$ \lceil \frac{W_t}{c'_{j+1}}\rceil-1 + (x-x_t)\tau.$$
 Since this is true for every $t$ we have proven one direction of the lemma, i.e.,
 $$
 \Theta_L(P,x) \ge  
\max_{0 \le t \le r}
\left(\left((x-x_{t})\tau+\left\lceil \frac{W_t} {c'_{t+1}}\right\rceil -1\right)\right).
$$
 
 To prove the other direction let $x_t$ be the last vertex at which $L$ waits before reaching $x.$  If $L$ never waits, set $x_t=x_0,$ the original location of $L.$   Note that, since $L$ never waits after $x_t$, it takes $L$ exactly $(x-x_t)\tau$ time to reach $x$ from $x_t.$
 
 We now claim that, at every time step, $s=0,\ldots T_t$, except possibly the last one, exactly $c'_{t+1}$ items pass through  $x_i$,  This is obviously true in the case $t=0.$ Suppose this was not true if  $t >0$.

Suppose there was ever a time $s$ at which no people were waiting at $x_t$.
Since $c'_{t} \le c'_{t+1}$, this implies any wave arriving at $x_t$ from  $e_t$ will pass right through to $e_{t+1}$ without waiting, leaving no people waiting at $x_t$ at time $s+1$.  This in turn implies that if there were ever a time $s$ at which no people were waiting at $x_t$, no one would ever wait at $x_t$ at  a later time step.  Recall that $x_t$ was chosen as the last vertex at which $L$ waits.  Thus at  {\em every} time $s,$  $0 \le s < T_t$, someone must be waiting on $x_t.$ People can  only wait on $x_t$ if the wave leaving $x_t$ at time $s$ is full, i.e., contains $c'_{t+1}$ items.  We have therefore seen that or all times $s,$ $0 \le s < T_t$, $c'_{t+1}$ items leave $x_t$
and therefore  
$T_t = \lceil W_t/c'_{t+1}\rceil -1$ so
$$\Theta_L(P',x) = \lceil W_t/c'_{i+1}\rceil-1 + (x-x_i)\tau.$$
 proving the other direction  of the lemma.   
\end{proof}

\medskip

By symmetry a  similar expression exists for the evacuation time to the right of sink $x$.  We have thus just 
proven Theorem  \ref {thm:main form} and, by extension, Claim \ref {claim:theta1PXalg}.

\section{The 1-sink Evacuation Algorithm}
\label{sec:1_sink_evac}

In what follows we often implicitly assume that $w_0>0$ and $w_n>0$.   If this is not the case, then   preprocess in linear time by clipping off the leftmost and rightmost vertices that have zero weight.

Note that if sink $x$ is required to be one of the original $x_i$  then $\Theta^1(P)$ is well defined as the minimum over $n+1$ values.   But, if $x \in P$ is a non-vertex edge point, then it is not a-priori obvious that such a minimum exists.  We now develop some basic properties that ensure its existence and lead to the efficient calculation algorithm claimed by Theorem \ref {thm:1sink}.

\begin{figure}[t]
\begin{center}
\includegraphics[width=2.5in]{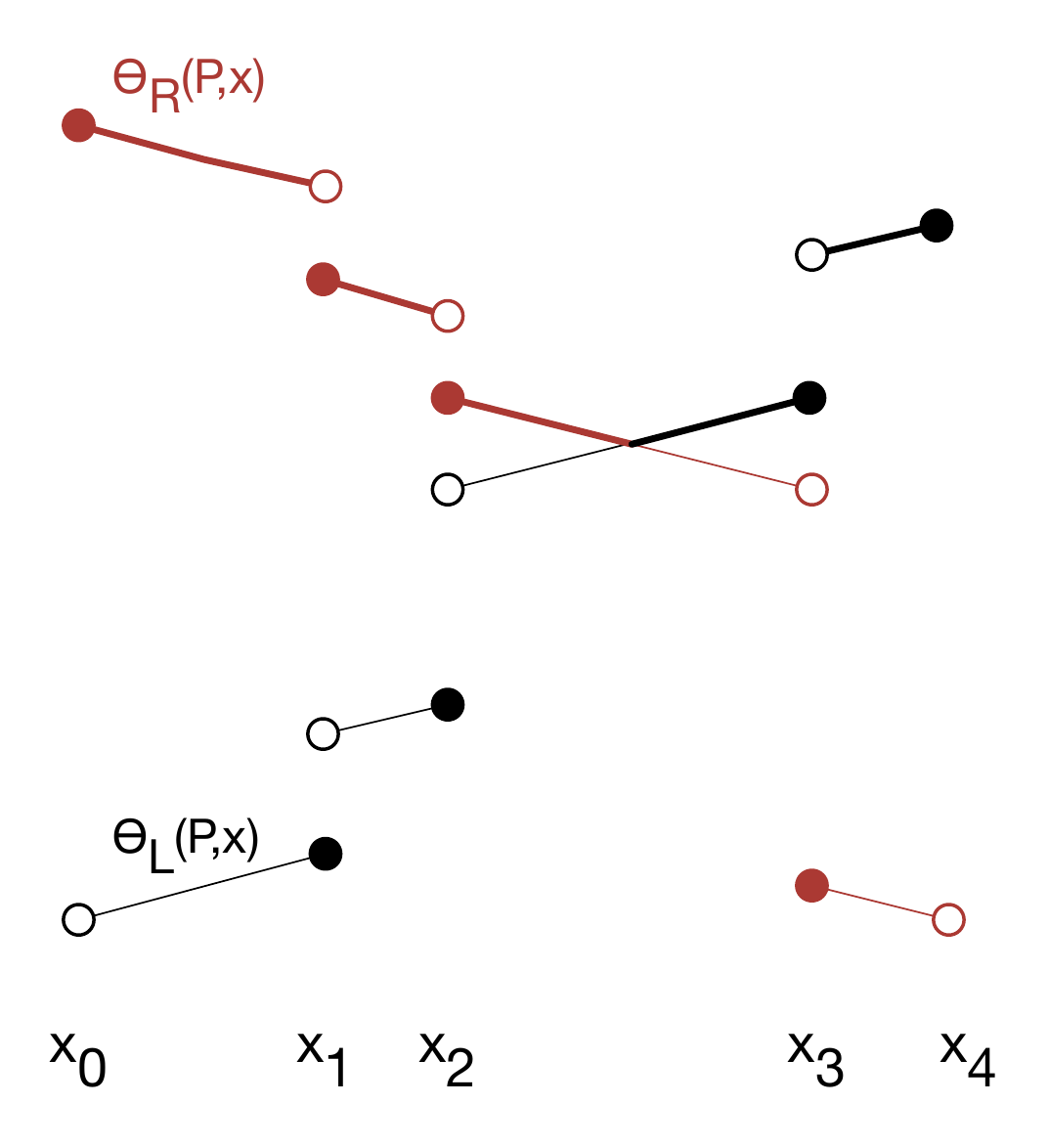}
\end{center}
\caption{An illustration of Claim \ref{claim:mono}. $\Theta_L(P,x)$ is a monotonically increasing  piecewise linear function in $x$;  $\Theta_R(P,x)$ is a monotonically decreasing  piecewise linear function in $x$.  $\Theta(P,x)$ (the heavy lines) is the max of the two functions and is unimodal.  $\Theta(P) = \min_x \Theta(P,x)$ is the unique value at which $\Theta(P)$ achieves its minimum.}
\label{fig:Theta_2_a}
\end{figure}

It is straightforward to see that
\begin{Claim}＼　
\begin{itemize}
\item $\Theta_L(P,x_0) =0 =  \Theta_R(P,x_n)$
\item $\Theta_L(P,x)$ is monotonically increasing and $\Theta_L(P,x)$ is monotonically decreasing for  $x \in P$.
\end{itemize}
\end{Claim}

Consider  $x_i < x < x_{i+1}$.  $\Theta_{L}(P,x_{i+1})$ is the time at which the wave containing the final item from $x_0$ reaches $x_{i+1}$.  Since waves travel at speed $1 /\tau$,   that final item from $x_0$ had reached location $x$,  at time 
$\Theta_{L}(P,x_{i+1}) - \tau (x_{i+1} -x)$.   This justifies (Figure \ref{fig:Theta_2_a})
\begin{Claim}＼　
\label{claim:mono}
\begin{itemize}
\item If $x_i < x < x_{i+1}$ then
\begin{eqnarray}
\Theta_{L}(P,x)   &=& \Theta_{L}(P,x_{i+1}) - \tau (x_{i+1} -x) \label{eq:Lline}\\
\Theta_{R}(P,x)   &=& \Theta_{R}(P,x_{i}) - \tau (x-x_i)  \label{eq:Rline}
\end{eqnarray}
\item $\Theta_{L}(P,x)$ is a monotonically increasing piecewise linear function with slope $\tau$ and  (possible) discontinuities at $x=x_i$ at which it is left continuous.
\item $\Theta_{R}(P,x)$ is a monotonically decreasing  piecewise linear function with slope $-\tau$ and  (possible) discontinuities at $x=x_i$ at which it is right continuous.
\end{itemize}
\end{Claim}

\begin{figure*}[tp]
\begin{subfigure}[b]{1.5in}
\includegraphics[width=1.5in]{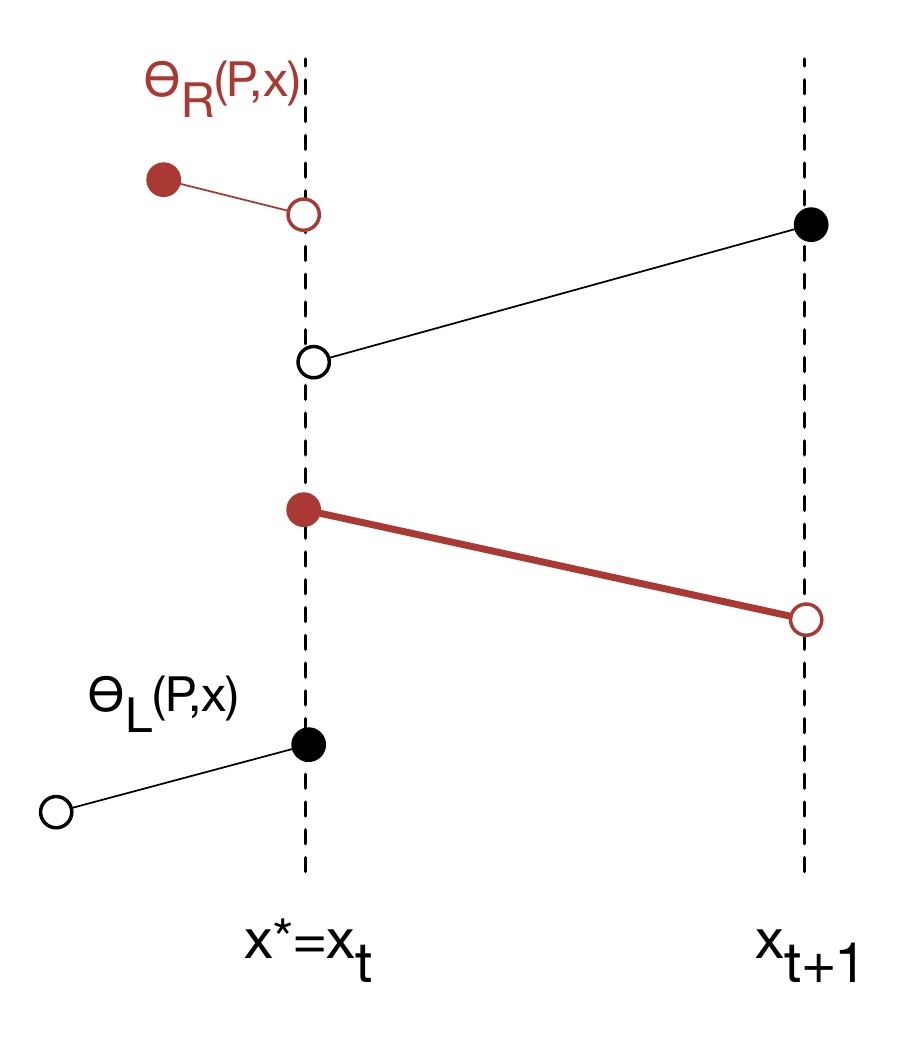}
\caption{\label{fig:x*:Case1}  Case 1}
\end{subfigure}
\hfill
\begin{subfigure}[b]{1.5in}
\includegraphics[width=1.5in]{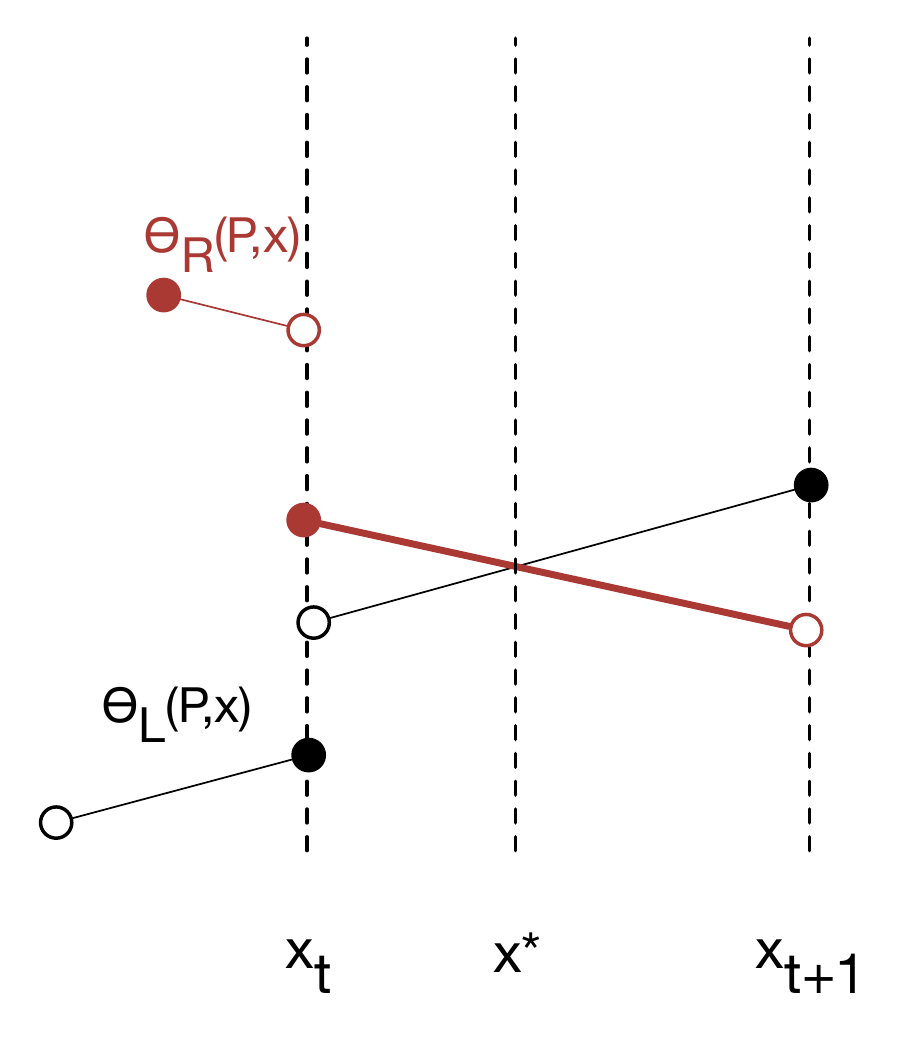}
\caption{\label{fig:x*:Case2} Case 2}
\end{subfigure}
\hfill
\begin{subfigure}[b]{1.5in}
\includegraphics[width=1.5in]{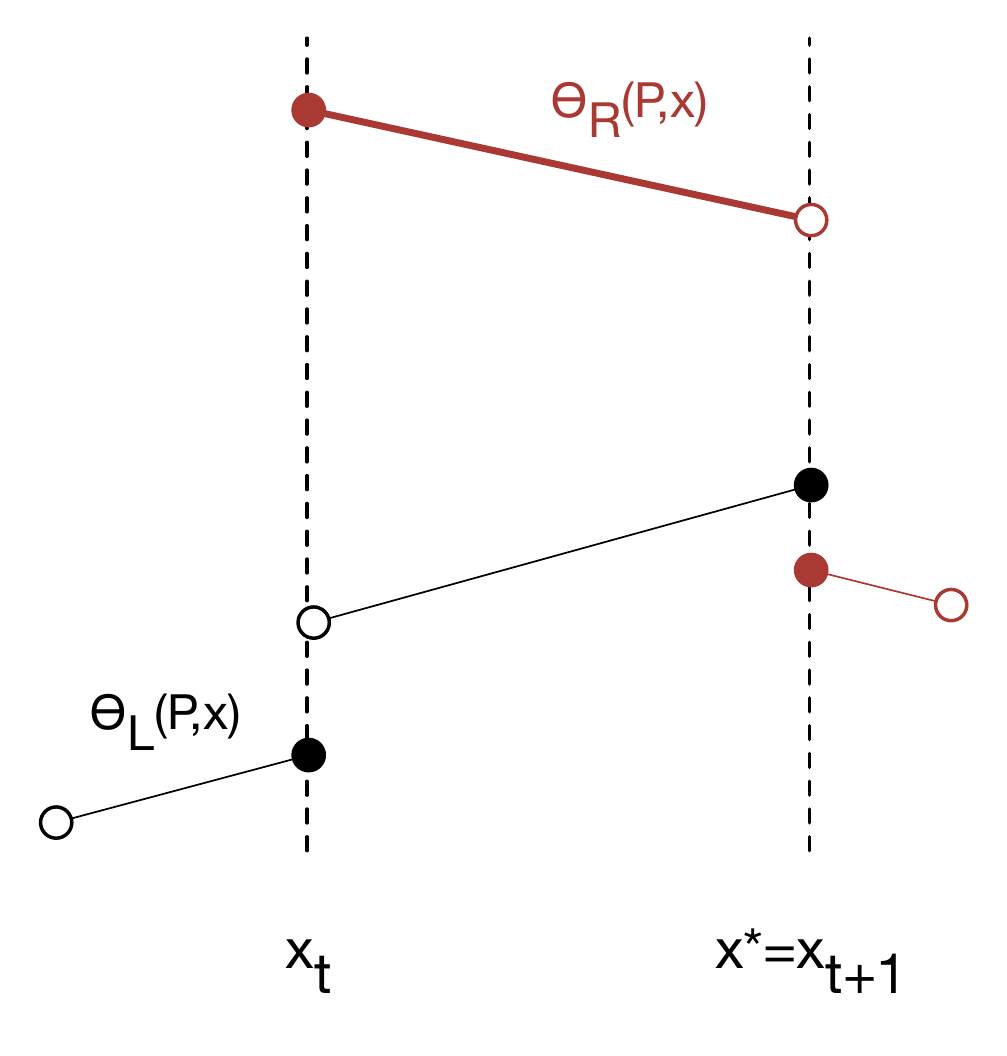}
\caption{\label{fig:x*:Case3} Case 3}
\end{subfigure}
\caption{\label{fig:2x*} The three cases from  Lemma \ref{lem: 1opt}'s proof.  $\Theta_L$ and $\Theta_R$ are piecewise linear functions with their (possible) discontinuities restricted to the $x_t$. $\Theta_L$ is left-continuous at every point and $\Theta_R$ is right-continuous at every point. By the definition of $t$, $\Theta_L(P,x_t) <  \Theta_R(P,x_t)$ and  $\Theta_L(P,x_{t+1}) \ge \Theta_R(P,x_{t+1})$. }
\end{figure*}

Combining the two previous claims  yields
\begin{Claim}
\label{claim:unimodality}
\item $\Theta^1(P,x)$ is a {\em strictly unimodal} function of $x$, i.e., it monotonically decreases, achieves its minimum at a unique value $x*$ and then monotonically increases.
\item Furthermore $x^*$ is the unique value for which
\begin{eqnarray*}
\forall x < x^*,  & \Theta(P,x)  = \Theta_R(P,x) >  \Theta_R(P,x)\\
\forall x > x^*,  & \Theta(P,x)  = \Theta_L(P,x) > \Theta_L(P,x)
\end{eqnarray*}
\end{Claim}

It is now straightforward to derive
\begin{lemma}
\label{lem: 1opt}
Let $t= \max \{ i \,:\, \Theta_L(P,x_i) <  \Theta_R(P,x_i)\}$.
Set 
\begin{eqnarray*}
\Theta^+_L(P,x_t) &=& \lim_{x\downarrow x_t} \Theta_L(P,x) = \Theta_{L}(P,x_{t+1}) - \tau (x_{t+1} -x_t)\\
\Theta^-_R(P,x_{t+1}) &=& \lim_{x\uparrow x_{t+1}} \Theta_R(P,x) = \Theta_{R}(P,x_t) - \tau (x_{t+1}-x_t)
\end{eqnarray*}
Then $\Theta^1(P) = \Theta(P,x^*)$ where
$$\small
x^* =
\left\{
\begin{array}{ll}
x_t  & \mbox{if $\Theta^+_L(P,x_t) \ge \Theta_R(P,x_t)$}\\
\frac 1 {2 \tau} (\Theta_L(P,x_{t+1}) - \Theta_R(P,x_{t}) - \tau(x_{t+1}-x_t)) & \mbox{if $\Theta^+_L(P,x_t) <  \Theta_R(P,x_t)$ and  $\Theta_L(P,x_{t+1}) > \Theta^-_R(P,x_{t+1})$    }\\
x_{t+1}  & \mbox{if   $\Theta_L(P,x_{t+1}) \le \Theta^-_R(P,x_{t+1})$    }
\end{array}
\right.
$$
\end{lemma}

\begin{proof}  See Figure \ref{fig:2x*}.

\par\noindent \underline{1: $\Theta^+_L(P,x_t) \ge \Theta_R(P,x_t)$:}\\
If $x \in (x_t,x_{t+1})$  then 
\begin{eqnarray*}
\Theta_L(P,x) &=& \Theta^+_L(P,x_t) + \tau(x-x_t)\\
				   &>& \Theta_R(P,x_t) - \tau(x-x_t) = \Theta_R(P,x)
\end{eqnarray*}
Thus, by the definition of $t$ and  Claim \ref {claim:unimodality}, $x* = x_t $ and
$\Theta^1(P) = \Theta_R(P,x^*).$

\medskip

\par\noindent \underline {2: $\Theta^+_L(P,x_t) <  \Theta_R(P,x_t)$ and  $\Theta_L(P,x_{t+1}) > \Theta^-_R(P,x_{t+1})$ 
}\\
(\ref{eq:Lline}) and (\ref{eq:Rline}) define  $\Theta_L(P,x)$ and $\Theta_R(P,x)$ for
$x \in (x_t,x_{t+1})$.  Consider them as lines defined on the closure $[x_t,x_{t+1}]$.  Note that the condition
states that $\Theta_L(P,x)$ starts  below  $\Theta_R(P,x)$ and ends above it.  The lines therefore cross in the interval.  Calculation yields this is at
$$x' = \frac 1 {2 \tau} (\Theta_L(P,x_{t+1}) - \Theta_R(P,x_{t}) - \tau(x_{t+1}-x_t)).$$
By definition $\Theta_L(P,x') = \Theta_R(P,x')$ so, by claim \ref {claim:unimodality}, $x* = x'$ 
and
$\Theta^1(P) = \Theta_L(P,x^*)  = \Theta_R(P,x^*).$

\medskip

\par\noindent \underline {3: $\Theta_L(P,x_{t+1}) \le \Theta^-_R(P,x_{t+1})$:}\\

If $x \in (x_t,x_{t+1})$  then 
\begin{eqnarray*}
\Theta_L(P,x) &=& \Theta_L(P,x_{t+1}) - \tau(x_{t+1}-x)\\
				   &<& \Theta_R^-(P,x_t) +  \tau(x_{t+1}-x) = \Theta_R(P,x)
\end{eqnarray*}
Thus, by the definition of $t$ and  Claim \ref {claim:unimodality}, $x* = x_{t+1} $ and
$\Theta^1(P) = \Theta_L(P,x^*).$

\end{proof}

\medskip

The proof of Theorem \ref{thm:1sink} now follows:
Claim  \ref{claim:unimodality} allows  binary searching the {\em vertices} of $P$ to find 
$$t= \max \{ i \,:\, \Theta_L(P,x_i) <  \Theta_R(P,x_i)\}.$$
  At each vertex $x_i$, it is only necessary to check  if $\Theta_L(P,x_i) >  \Theta_R(P,x_i)$ which, from Claim \ref{claim:theta1PXalg}, requires $O(|P|)$ time.  Thus, $t$ can be found in  $O(|P| \log |P|)$ time.

After the binary search completes, 
 $\Theta_L(P,x_t)$, $\Theta_R(P,x_t)$,   $\Theta_L(P,x_{t+1})$ and $\Theta_R(P,x_{t+1})$ are all known.  Lemma
\ref {lem: 1opt} gives  the sink $x^*$ and value $\Theta^1(P) = \Theta(P,x)$ in another $O(1)$ time.  The total time taken by the two steps is $(|P| \log |P|)$.


\section{$k$-Sink Evacuation}
\label{sec:k_sink_evac}

\subsection{Properties}
\label{subsec:k_sink_evac_Prop}
From Definition \ref {def:ksink1} it  is easy  to derive

\begin{Claim}
\label{claim:min_k}
\begin{eqnarray*}
\forall k> 1,\quad  \Theta^k(P)  &=& \min_{\mbox{\rm partition}\  I}\  \max_{1 \le j \le k} \Theta^1(P^j_I)\\
				  &=& \min_{0 \le i < n} \max\left(\Theta^1(P_{0,i}),\, \Theta^{k-1}(P_{i+1,n})\right).
\end{eqnarray*}
\end{Claim}

We will also need the following straightforward facts which simply state that adding a vertex, along with its associated weight and edge, to either the left or right side of a path, can not decrease the $k$-sink evacuation times.

\begin{Claim}
\label{claim:unimodal_k}
$$\forall i\le j < n,\,  \forall k, \quad  \Theta^k(P_{i,j}) \le \Theta^k(P_{i-1,j})$$
and
$$\forall 0< i\le j\,  \forall k, \quad  \Theta^k(P_{i,j}) \le \Theta^k(P_{i-1,j})$$
\end{Claim}

\subsection{Testing}
\label{subsec:k_sink_evac_Test}
Before  calculating $\Theta^k(P)$ we first show how to test  whether $\Theta^k(P_{i,j}) \le \alpha$, for any specified $\alpha>0.$
This uses the following direct consequence of Claim \ref{claim:unimodal_k}.
\begin{Claim}
\label{claim:test_recur}
Set 
$$j^* = \max \{   j   \,:\,   i \le  j \le n  \mbox{ and }  \Theta^1(P_{i,j}) \le \alpha \}.$$
Then, for $k>1,$ 
$$
\Theta^k(P_{i,n}) \le \alpha
\quad \quad\mbox{\rm if and only if}\quad \quad 
j^*=n  \ \mbox{\rm or } \ \Theta^{k-1}(P_{j^*+1,n}) \le \alpha
$$
\end{Claim}

\begin{algorithm}[t]

\caption{{\em Peel} returns  {\bf Yes}  if $\Theta^k(P_{i,n} \le \alpha)$.  Otherwise it returns {\bf No}.}
\label{alg:test}
\begin{algorithmic}[1]
\Procedure{Peel}{$P_{i,n},\,k,\, \alpha$}
\State {$j^* = \max \{   j   \,:\,   i \le  j \le n   \mbox{ and }  \Theta^1(P_{i,j}) \le \alpha \}$}
\If  {$j^*= n$}
       \State {Return(Yes) }   
\ElsIf  {$k=1$}
      \State {Return(No) }   
\Else
       \State {Return({Test}{($P_{j^*+1,n},\,k-1,\, \alpha$)}) }  
  \EndIf         
\EndProcedure
\end{algorithmic}
\end{algorithm}

This immediately allows checking   $\Theta^k(P_{i,n}) \le \alpha$ via the {\em peeling process} described in Algorithm \ref {alg:test}: Peel $P_{i,j^*},$ the largest possible subpath that can be evacuated in time $\le \alpha$, 
 off from the left side of $P_{i,n}$  by replacing  $i$ with $j^*$. 
 
 If $j^*=n$ then $\Theta^k(P_{i,n}) \le \alpha$. 
 
 If $j^*\not=n$ and $k=1$ then $\Theta^k(P_{i,j^*})> \alpha$.  Otherwise, $\Theta^k(P_{i,n}) \le \alpha$ if and only if
 $\Theta^{k-1}(P_{j^*+1,n}) \le \alpha$

%


%

Claim  \ref{claim:j*calc} will permit finding   $j^*$ in  $O(|P_{i,j^*}|\log |P_{i,j^*}|) $ time.
Thus,  the cost of peeling off a subpath $P'$ is $O(|P'| \log |P'|)$. 
Suppose that the algorithm iteratively peels off subpaths $P'_1,P'_2,\ldots,P'_{k'}$ from $P_{i,n}$ where $k' \le k.$
Once  a subpath is peeled off it is never put back so  $\sum_{t \le k'} |P'_t| \le |P_{i,n}|$ and
 the total cost of Algorithm \ref {alg:test}  is thus
$$\sum_{t \le k'} |P'_t|\log |P'_t| \le \log |P_{i,n}|  \sum_{t \le k'} |P'_t|\log |P'_t| \le |P_{i,n}| \log |P_{i,n}|.$$

To prove Theorem \ref {thm:ksink_ver} it therefore suffices to show how to find $j^*$ in
$O(|P_{i,j^*}|\log |P_{i,j^*}|)$ time.
This id done by first scanning right from $i$
to find the largest $x^*$ such that $\Theta_L(P_{i,j},x) \le \alpha$ and then scanning right from $x^*$ to find  the largest $j^*$ such that $\Theta_R(P_{i,j},x^*) \le \alpha.$ The correctness of this procedure is proven below.
(See Figure \ref{fig:doubling})
\begin{Claim}
\label{claim:peel}
Let $i$ be fixed and $\alpha >0$. Define
\begin{eqnarray*}
x^* &=&  \max \{  x  \in P_{i,n}\,:\,   x_i \le  x   \mbox{ and }  \Theta_L(P_{i,n},x) \le \alpha \}\\
j^* &=&  \max \{   j   \,:\,   x^* \le  x_j   \mbox{ and }  \Theta_R(P_{i,j},x^*) \le \alpha \}.
\end{eqnarray*}
where $x* \in P_{i,n}$ and $j^*$ is a vertex in $P_{i,n}$.
Then 
$$j^* = \max \{   j   \,:\,   i \le  j \le n  \mbox{ and }  \Theta^1(P_{i,j}) \le \alpha \}.$$
Furthermore,
$$\Theta(P_{i,j^*},x^*) \le \alpha.$$
\end{Claim}
\begin{proof}
Let $$j'  = \max \{   j   \,:\,   i \le  j   \mbox{ and }  \Theta_R(P_{i,j},x) \le \alpha \}.$$ 
Suppose the claim is false, i.e., $j^* \not=j'$ so $j^* < j'.$
Let $x'$ be the unique sink satisfying 
$\Theta^1(P_{i,j'}) = \Theta(P_{i,j'},x').$

First note that 
$$\Theta_L(P_{i,n},x') = \Theta_L(P_{i,j},x') \le \Theta^1(P_{i,j},x') \le \alpha$$
so, by definition,  $x' \le x^*.$  By assumption
 $x^* \le j^* <  j'$ so $\Theta_R(P_{i,j'},x^*)$ is well defined and, by the decreasing monotonicity of $\Theta_R(P_{i,j'},x)$ from   Claim \ref{claim:mono}
$$\Theta_R(P_{i,j'},x^*) \le \Theta_R(P_{i,j'},x')\le \Theta^1(P_{i,j'},x')  \le \alpha.$$
But, from the definition of $j^*$, this immediately implies $j' \le j^*$ contradicting the original assumption.  Thus $j' = j.$

The fact that $\Theta(P_{i,j^*},x^*) \le \alpha$ follows directly from the definitions of $x^*,j^*.$

\end{proof}

\begin{figure}
\begin{center}
\includegraphics[width=4.5in]{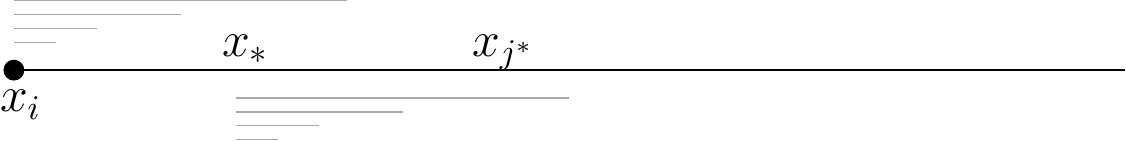}
\end{center}
\caption{An illustration for Claims  \ref{claim:peel} and \ref{claim:j*calc}.  $x^*$ (not necessarily a vertex) is the  rightmost point at which 
$\Theta_L(P_{i,n},x^*) \le \alpha$.  $j^*$ is the rightmost point at which $\Theta_R(P_{i,j^*},x^*) \le \alpha$.
$x^*$ is found by doing a doubling search on $i'-i$, and then a binary search in $[i,i']$ to find the largest $i'$ such that 
$\Theta_L(P_{i,i'},x_{i'}) \le \alpha$ and then shifting slightly to the right. $j*$ (which is a vertex) is found using a similar doubling search and binary search technique. At the procedures' end, $j^*$ is the largest index such that $\Theta^1(P_{i,j^*}) \le \alpha$. }
\label{fig:doubling}
\end{figure}

Finally, $x^*$ and $j^*$ can be found efficiently. 
\begin{Claim}
\label{claim:j*calc}
Let $i$ be fixed. Then  $x^*,j^*$ as defined by Claim \ref{claim:peel} can be found in 
$O(|P_{i,j^*}|  \log |P_{i,j^*}|)$ time.
\end{Claim}
\begin{proof}

Use doubling search.  For given  $x_{i'}$, Claim  \ref{claim:theta1PXalg}  permits testing in $O(|P_{i,i'}|)$ time whether 
 $\Theta_L(P_{i,i'},x_{i'}) \le \alpha$.   Now
\begin{enumerate}
\item Iteratively test $k=0,1,2,\ldots$ until finding the smallest  $k$ such that $\Theta_L(P_{i,i'},x_{i'}) > \alpha$ for
$i'=i+2^k$. Label this as $k^*$.
\item Binary search  in the range  $[ i+2^{k^*-1},i+2^{k^*})$ for the largest $i'$ such that $\Theta_L(P_{i,i'},x_{i'}) \le \alpha$
\item Use Claim \ref{claim:mono} to find $x^* \in [x_{i'}, x_{i'+1}]$
\end{enumerate}
Note that $i'$ found in step 2 satisfies $i'\le j^*$ so, by the properties of the doubling search.
$$2^{k^*} \le 2(i'-i+i)  \le 2 (j^*-i+1) = |P_{i,j^*}|$$
Thus, the total number of tests performed in steps (1) and (2) is $O(\log |P_{i,j^*}|)$.  Since each test requires 
$O(|P_{i,j^*}|)$ time the total time for (1) and (2) is $O(P_{i,j^*}| \log |P_{i,j^*}|)$.
(3) only requires $O(1)$ time so finding $x^*$ only requires $O(P_{i,j^*}| \log |P_{i,j^*}|)$ in total.

Similarly, once $x^*$ is known a  doubling search can be used to find the largest index $j^*$ such that
$\Theta_R(P_{i,j^*}, x^*) \le\alpha$. By a similar analysis the cost if this doubling search will be $O(|P_{i',j^*}|\log |P_{i',j^*}|)$,
so the total cost of finding both $x^*$ and $j^*$ is $O(|P_{i,j^*}|  \log |P_{i,j^*}|)$.
\end{proof}

We note that the algorithm described so far only checks whether $\Theta^k(P_{i,j}) \le \alpha$ is true or not. The statement of Theorem \ref{thm:ksink_ver} requires that if the statement is true, then  the algorithm must also provide an evacuation protocol with time $\le \alpha$. Such a protocol is a partition into subintervals with a sink for each subinterval such that the evacuation time for each subinterval to its designated sink is $\le \alpha$.  

Algorithm \ref{alg:test} can be easily modified to return such a protocol without increasing its asymptotic running time.    Algorithm  \ref{alg:test} uses Claims \ref {claim:peel} and \ref{claim:j*calc} to find the $j^*$ with the $j^*$ denoting the end of the current peeled subpath.  When finding $j^*$, the subroutine from Claim \ref{claim:j*calc} also explicitly finds the sink $x^*$ to which the current subpath  evacuates in time $\le \alpha$.  To return the required evacuation protocol it is therefore only necessary to store the list of the $j^*,x^*$ generated while performing the peeling and have the algorithm return them if $\Theta^k(P_{i,j}) \le \alpha$.

\vspace*{-.04in}

\subsection{The Final Algorithm}
\label{subsec:k_sink_evac_Alg}

\begin{figure}
\begin{center}
\includegraphics[width=4.5in]{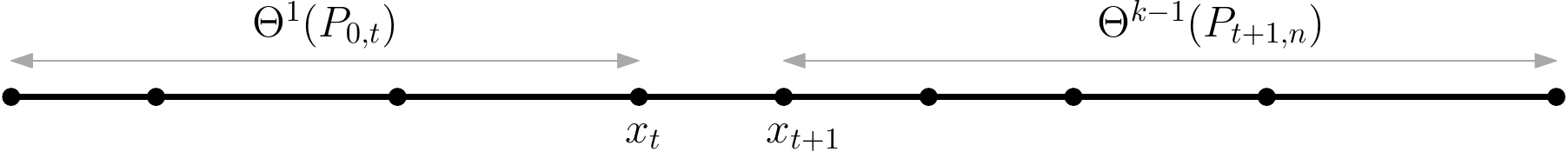}
\end{center}
\caption{Illustration for Theorem  \ref{thm:ksink_main}.
$t$ is the largest index such that $\Theta^1(P_{0,t}) < \Theta^{k-1}(P_{t+1,n})$.  Then,
$\Theta^k(P) = \min  \left(\Theta^1(P_{0,{t+1}}),\, \Theta^{k-1}(P_{t+1,n})  \right)$.} 
\label{fig:final_alg}
\end{figure}

Our  final algorithm for  calculating $\Theta^k(P)$ assumes that at least $k+1$ vertices on $P$ have positive weight.  If not, then the problem is trivial since each such vertex could be assigned its own sink, leading to an optimal $0$ evacuation time.

Under  that assumption,
 first observe that Claims   \ref {claim:min_k}  and  \ref{claim:unimodal_k} directly imply (See Figure \ref{fig:final_alg})

\begin{theorem}
\label{thm:ksink_main}
For $k>1$ such that  $1 < k < n$ set 
\vspace*{-.05in}
$$t = \max\{i \,:\, \Theta^1(P_{0,i}) < \Theta^{k-1}(P_{i+1,n})\}$$
\vspace*{-.03in}
Then 
\begin{itemize}
\item[\bf (a)] $0 \le t < n$
\item[\bf (b)]  $\forall i \le t,\ \Theta^1(P_{0,i}) < \Theta^{k-1}(P_{i+1,n})$
\item[\bf (c)]  $\forall i > t,\ \Theta^1(P_{0,i}) \ge \Theta^{k-1}(P_{i+1,n})$
\vspace*{.03in}
\item[\bf (d)]
$\Theta^k(P) = \min  \left(\Theta^1(P_{0,{t+1}}),\, \Theta^{k-1}(P_{t+1,n})  \right)$
\end{itemize}

\end{theorem}

\begin{proof}

\ \\

\par \noindent{\bf (a)} $t \ge 0$ because $\Theta^1(P_{0,0})=0$ while $\Theta^{k-1}(P_{i+1,n}) >0.$
Similarly, $t <n$ because $\Theta^{k-1}(P_{n,n}) =0.$

\par \noindent{\bf (b)} If $i=t$ the claim follows from the definition of $t.$  
If $i < t$ then from the definition of $t$ and repeated applications of Claim  \ref{claim:unimodal_k}
$$ \Theta^1(P_{0,i}) \le \Theta^1(P_{0,t})  <  \Theta^{k-1}(P_{t+1,n}) \le \Theta^{k-1}(P_{i+1,n})
$$

\par \noindent {\bf (c)}  From the definition of $t$,
$\Theta^1(P_{0,t+1})  \ge  \Theta^{k-1}(P_{t+2,n})$. If $i >t$ then repeated applications of Claim  \ref{claim:unimodal_k} yield
$$ \Theta^1(P_{0,i}) \ge \Theta^1(P_{0,t})  \ge  \Theta^{k-1}(P_{t+1,n}) \ge \Theta^{k-1}(P_{i+1,n}).$$

\par \noindent  {\bf (d)}  If $i \le t$ then,  from {\bf (b)}, 
$$\max\left(\Theta^1(P_{0,i})\, \Theta^{k-1}(P_{i+1,n})\right)
= \Theta^{k-1}(P_{i+1,n}) > \Theta^{k-1}(P_{t+1,n}).
$$
If $i > t$ then, from {\bf (c)},
$$\max\left(\Theta^1(P_{0,i})\, \Theta^{k-1}(P_{i+1,n})\right)
=\Theta^1(P_{0,i}) \ge \Theta^1(P_{0,t+1}).$$

Plugging this into    Claim \ref {claim:min_k} yields
\begin{eqnarray*}
  \Theta^k(P)   &=& \min_{0 \le i < n} \max\left(\Theta^1(P_{0,i}),\, \Theta^{k-1}(P_{i+1,n})\right)\\
					 &\ge&\min  \left(\Theta^1(P_{0,{t+1}},\, \Theta^{k-1}(P_{t+1,i})  \right).
\end{eqnarray*}
From the other direction applying {\bf (b)}  and {\bf (c)} with $i=t,t+1$ yields
\begin{eqnarray*}
\max\left(\Theta^1(P_{0,t})\, \Theta^{k-1}(P_{t+1,n})\right) &=& \Theta^{k-1}(P_{t+1,n})\\
\max\left(\Theta^1(P_{0,t+1})\, \Theta^{k-1}(P_{t+2,n})\right) &=&\Theta^1(P_{0,t+1})
\end{eqnarray*}
implying
$$ \Theta^k(P)  =  \min  \left(\Theta^1(P_{0,{t+1}}),\, \Theta^{k-1}(P_{t+1,n})  \right),$$
proving  {\bf (d)}.
\end{proof}

\medskip

 Algorithm \ref{alg:OPTk} provides an upper level sketch of how to calculate 
$\Theta^k(P_{j,n})$, for any $j \ge 0.$ Its correctness follows immediately from 
 Theorem \ref {thm:ksink_main}.
 
\begin{algorithm}[t]

\caption{Calculates $\Theta^k(P_{j,n})$ .}
\label{alg:OPTk}
\begin{algorithmic}[1]
\Procedure{OPT}{$P_{j,n},\,k$}
\If  {$k= 1$}
       \State { {\bf Return}$\Bigl(\Theta^1(P_{j,n})\Bigr)$ }   
\Else
       \State {Calculate $t = \max\{i \ge j \,:\, \Theta^1(P_{j,i}) < \Theta^{k-1}(P_{i+1,n})\}$ }  
       	\State {{\bf Return}$\Bigl(\min  \left(\Theta^1(P_{j,{t+1}}),\, \Theta^{k-1}(P_{t+1,n})  \right)\Bigr)$}
  \EndIf         
\EndProcedure
\end{algorithmic}
\end{algorithm}

 To analyze the running time of Algorithm \ref{alg:OPTk}, note that lines 2-3 can be implemented in time $O(|P_{j,n}| \log |P_{j,n}|)$ using Theorem  \ref {thm:1sink}.
 
 \medskip
 
 If $k >1,$ 
for any fixed $i$,  checking whether
$\Theta^{k-1}(P_{i+1,n}) \le  \Theta^1(P_{j,i})$ can be done in
 $O(|P_{j,n}| \log |P_{j,n}|)$ time as follows:
\begin{enumerate}
\item Calculate  $\alpha = \Theta^1(P_{j,i})$\\
       From Theorem  \ref {thm:1sink}, this takes $O(|P_{j,i}| \log |P_{j,i}|)$ time.
\item Check whether $\Theta^{k-1}(P_{i+1,n}) \le \alpha$\\
		From Theorem \ref {thm:ksink_ver} this takes $O(|P_{i+1,n}| \log |P_{i+1,n}|)$ time.
\end{enumerate}

\medskip

Theorem \ref {thm:ksink_main} (b),(c) implies that we can binary search over the indices to find $$t = \max\{i \ge j \,:\, \Theta^1(P_{j,i}) < \Theta^{k-1}(P_{i+1,n})\}.$$
The binary search  tests $O(\log |P_{j,n}|)$ vertices, so the total  time required for finding $t$ in line 5 of the algorithm is 
 $O(|P_{j,n}| \log^2 |P_{j,n}|)$ time.  Note that at the end of the binary search both $\Theta^1(P_{j,t})$ and $\Theta^1(P_{j,t+1})$ (if $t < n$) are known.

Line 9 needs to recursively call $OPT^{k-1}(P_{t+1,n})$. Since at most $k-1$ recursive calls are made, the entire implementation of  Algorithm \ref{alg:OPTk} requires $O(k|P_{j,n}| \log^2 |P_{j,n}|)$ time.

\medskip

As written, Algorithm \ref{alg:OPTk} only returns the {\em value} of $OPT^{k}(P_{j,n})$. Theorem \ref {thm:ksink_main}  also has to provide the associated evacuation protocol, i.e., the partition into subpaths with associated sinks.

Algorithm \ref{alg:OPTk} can be easily modified to provide this protocol.  If $k=1$, there is no partition and Theorem  \ref {thm:1sink}  provides the sink.

If $k>1$ set 
$$\alpha = \Theta^k(P_{j,n}) = \min  \left(\Theta^1(P_{j,{t+1}}),\, \Theta^{k-1}(P_{t+1,n}) \right).$$
There are two cases.

\medskip

$\bf \alpha = \Theta^1(P_{j,{t+1}}):$ Then $\Theta^{k-1}(P_{t+2,n}) \le \alpha$.  An optimal protocol is given by returning the concatenation of the (i) subpath $P_{j,t+1}$ with its associated optimal sink given by Theorem  \ref {thm:1sink} and (ii) the $k-1$ subpaths and associated sinks of $P_{t+2,n}$ given by Theorem \ref{thm:ksink_ver}.  All of these  $k-1$ subpaths evacuate to their sinks in $\le \alpha$ time.

\medskip

 $\bf \alpha = \Theta^{k-1}(P_{t+1,n}):$ Then $\Theta^1(P_{j,{t}}) \le \alpha$. Return the concatenation of the 
(i) subpath $P_{j,t}$ with its associated optimal sink given by Theorem  \ref {thm:1sink} and (ii) the $k-1$ subpaths and associated sinks of  $P_{t+1,n}$ given by the recursive call to $OPT(P_{t+1,n},k-1).$

\medskip

\
\section{Conclusion}

This paper derived an $O(k n \log^2 n)$ algorithm for solving the $k$-sink location problem on a Dynamic Path network with general edge capacities.

The obvious question is whether this can be improved.   \cite{Higashikawa2014b}  gives an $O(kn)$ algorithm for solving the $k$-sink on a path problem {\em for uniform capacity edges}, i.e.,  $\forall i,\ c_i = c$ for some $c >0.$   That algorithm was based on the restriction   to the uniform case of Claims \ref {claim:mono}, \ref {claim:unimodality}, \ref{claim:min_k} and \ref{claim:unimodal_k}. Unlike the algorithms in this paper, which are based on binary search,  their algorithm {\em pointer shifted};  it moved $2k-1$  pointers from  left to right on the path vertices, identifying locations where optimality occurred in subproblems. 
$k-1$ of the pointers denoted the separation between subpaths;  The other $k$ pointers denoted the location of the optimal sinks within each subpath. Judicious application of the claims allowed proving that no pointer ever moved backwards so only $O(nk)$ pointer moves were made in total.

Each pointer move necessitated updating the calculation of either (i) some  $\Theta_L(P_{i,j}, x_j)$ to  $\Theta_L(P_{i+1,j}, x_j)$ or  $\Theta_L(P_{i,j+1}, x_{j+1})$ or (ii) some  $\Theta_R(P_{i,j}, x_i)$ to $\Theta_R(P_{i+1,j}, x_{i+1})$ or $\Theta_R(P_{i,j+1}, x_i).$   \cite{Higashikawa2014b} used a data structure that supported these updates in $O(1)$ amortized time for {\em uniform capacity} edges.
  This led to a $O(nk)$ total time algorithm.

That approach could not be used here because equations
(\ref{eq:left-evac-gc}) and (\ref{eq:right-evac-gc}) do not obviously allow efficient updating of
$\Theta_L(P_{i,j}, x_j)$ and $\Theta_R(P_{i,j}, x_i)$ 
for general capacities.

One approach to improving the algorithms in this paper would be to somehow devise a new data structure allowing updating in $O(f(n))$ time for $f(n) = o(\log^2 n)$. This would immediately lead to a $O(kn f(n))=o(kn f(n))$ algorithm for solving the general problem. This would seem to require a new approach to calculating evacuation times on a path.

\bibliographystyle{plain} 
\bibliography{evac}
\end{document}